\documentclass[%
 reprint,
 amsmath,amssymb,
 aps,
]{revtex4-2}

\DeclareMathOperator{\negl}{negl}

\usepackage{amsthm}
\newtheorem{theorem}{Theorem}

\newtheorem{definition}{Definition}

\newtheorem*{Game*}{Game}
\newtheorem*{TokGen*}{TokenGeneration Phase}
\newtheorem*{TokVer*}{TokenVerification Phase}
\newtheorem*{Setup*}{Setup}
\newtheorem*{Query*}{Query}
\newtheorem*{Challenge*}{Challenge}
\newtheorem*{Guess*}{Guess}
\usepackage{breqn}
\theoremstyle{remark}
\newtheorem{remark}[theorem]{Remark}

\usepackage{algorithm}
\usepackage{algpseudocode}

\usepackage{hyperref}
\hypersetup{
    hidelinks = true
}

\usepackage{subcaption}

\usepackage{comment}

\usepackage{tikz}
\usetikzlibrary{quantikz}
\usepackage{pgfplots}
\pgfplotsset{compat=1.18}
\usepgfplotslibrary{colormaps}
\usetikzlibrary{3d}
\usetikzlibrary{shapes,3d,calc}
\usepackage{amsmath}

\usepackage[capitalize]{cleveref}
\usepackage{tensor}
\usepackage{algpseudocode}
\usepackage{mathtools,physics, bbm}
\usepackage{graphicx}
\usepackage{dcolumn}
\usepackage{bm}

\begin{document}

\title{Quantum Physical Unclonable Function based on Chaotic Hamiltonians}

\author{$^{*}$Soham Ghosh, $^{*}$Holger Boche, $^{\dagger}$Marc Geitz}%
\affiliation{$^{*}$Chair of Information Technology$,$ Technical University of Munich$,$ Germany \\
 $^{\dagger}$T-Labs$,$ Deutsche Telekom AG$,$ Germany
}

\date{\today}

\begin{abstract}
Quantum Physical Unclonable Functions (QPUFs) are hardware-based cryptographic primitives with strong theoretical security. This security stems from their modeling as Haar-random unitaries. However, implementing such unitaries on Intermediate-Scale Quantum devices is challenging due to exponential simulation complexity. Previous work tackled this using pseudo-random unitary designs but only under limited adversarial models with only black-box query access. 

In this paper, we propose a new QPUF construction based on chaotic quantum dynamics. We modeled the QPUF as a unitary time evolution under a chaotic Hamiltonian and proved that this approach offers security comparable to Haar-random unitaries. Intuitively, we show that while chaotic dynamics generate less randomness than ideal Haar unitaries, the randomness is still sufficient to make the QPUF unclonable in polynomial time. Moreover, we show that the evolution time required to achieve security scales linearly with number of qudits used in the scheme and can be kept public.

We identified the Sachdev-Ye-Kitaev (SYK) model as a candidate for the QPUF Hamiltonian. 
Recent experiments using nuclear spins and cold atoms have shown progress toward achieving this goal. Inspired by recent experimental advances, we present a schematic architecture for realizing our proposed QPUF device based on optical Kagome Lattice with disorder. For adversaries with only query access, we also introduce an efficiently simulable pseudo-chaotic QPUF. Our results lay the preliminary groundwork for bridging the gap between theoretical security and the practical implementation of QPUFs for the first time.
\end{abstract}

\maketitle

\section{Introduction}

Quantum Physical Unclonable Functions (QPUFs) are a promising cryptographic primitive introduced in recent literature \cite{Skori2010,Arapinis2021quantumphysical,Ghosh2024}. In recent work \cite{Arapinis2021quantumphysical}, the hardware requirements for QPUF were formally defined. Modelling the QPUF as a Haar random unitary channel, several security guarantees were proved. Specifically, it was shown that no QPUF can remain secure against unbounded quantum adversaries running in Quantum Exponential Time (QET) (see \cref{Section: QPT}). However, against Quantum Polynomial Time (QPT) adversaries (see \cref{Section: QPT}), meaningful security notions were proved to be achieved. These include \emph{selective unforgeability} \cite{Arapinis2021quantumphysical} and \emph{measurement-selective unforgeability} \cite{Ghosh2024}. 

Importantly, QPUFs offer information-theoretic security under exact implementation of the protocols outlined in \cite{Arapinis2021quantumphysical, Ghosh2024} (see also \cref{Section: QPT} and \cref{SEC:III}). In contrast, classical PUFs (CPUFs) remain vulnerable to a range of attacks, including direct physical cloning \cite{helfmeier2013cloning}, and modeling attacks based on machine learning techniques \cite{ katzenbeisser2012pufs,ruhrmair2010modeling,ruhrmair2014efficient,ganji2015let,ganji2016strong,ganji2015attackers,ganji2016pac,santikellur2019deep}. Additionally, the quantum no-cloning theorem inherently protects the challenge-response pairs in QPUFs, providing another fundamental layer of security absent in classical counterparts.

However, exact implementations of Haar-random unitaries are infeasible on quantum devices due to exponential resource requirements, and no known physical system can realize them efficiently. As a result, QPUFs have remained largely theoretical.

To address this, prior work \cite{SWAP_Mina} introduced pseudo-Haar unitary $t$-designs that approximate Haar randomness efficiently on quantum devices. No computationally efficient algorithm with only black-box query access can distinguish a unitary sampled from a \( t \)-design from one sampled from the Haar measure. While this improves the practicality of QPUFs, the resulting constructions are only secure against QPT adversaries with access limited to black-box queries. For real-world cryptographic applications, broader adversarial models with more than query access are essential.

\subsection{Our Contributions}

In this work, we propose a novel, physically motivated QPUF construction based on chaotic quantum dynamics. We model the QPUF as the unitary time evolution generated by a chaotic Hamiltonian and prove that the resulting unitary satisfies the same security guarantees as Haar-random ones. Moreover, we show that the evolution time required to achieve security scales linearly with number of qudits and can be kept public. This approach is intuitive as the unclonability of the QPUF arises from the inherent randomness of chaotic dynamics. 

We further demonstrate that our chaotic QPUF model can be implemented using the Sachdev–Ye–Kitaev (SYK) Hamiltonian \cite{sachdevChaos,kitaeChaos}, a well-known model of quantum chaos with profound links to both black hole physics and condensed matter systems. While the SYK model features complex all-to-all interactions, recent experimental progress with nuclear spins \cite{SYK_Nuc} and ultracold atoms in optical lattices \cite{SYK_Optical} has shown that SYK-like Hamiltonians can be effectively realized using local, nearest-neighbour interactions. As a proof of concept, we present a schematic QPUF device architecture (see \cref{Sec:V}) based on a disordered Kagome lattice \cite{SYK_Optical}. Therefore, our construction opens the avenue for modeling attack scenarios beyond query access. 

Finally, as an alternative to pseudo-Haar random constructions, we propose a pseudo-chaotic QPUF construction, that is secure against adversaries with query-restricted access. The sampling procedure involves generating a unitary from a \( t \)-design as a subroutine, augmented by an additional layer of randomness through random eigenvalue sampling. An interesting open question is whether this added randomness yields any tangible advantage from using just pseudo-Haar unitaries as a model for QPUFs.

\subsection{Paper Organization}
This paper is organized as follows. In \cref{Sec:II} and \cref{SEC:III}, we briefly review the security notions of \emph{selective unforgeability}\cite{Arapinis2021quantumphysical} and \emph{measurement-selective unforgeability} \cite{Ghosh2024}. \cref{SEC:IV} presents our QPUF construction based on chaotic Hamiltonians and proves its security under both measures. \cref{Sec:V} and \cref{Sec:VI} describe two practical implementation approaches: a physical realization via SYK models and a simulative method using pseudo-chaotic Hamiltonians. Finally, \cref{Sec:VII} discusses the advantages and limitations of our proposals and outlines directions for future work.

\section{Brief Review of Selective Unforgeability Scheme \cite{Arapinis2021quantumphysical}}
\label{Sec:II}
The selective unforgeability security measure was introduced in \cite{Arapinis2021quantumphysical}. For clarity, we restate it in a form more convenient for our analysis. In this scheme, the QPUF is represented by a $D$-dimensional unitary channel $U$ drawn from the Haar measure over the unitary group. It also satisfies certain hardware requirements, namely, \emph{robustness} and \emph{collision-resistance}, which are described in \cite{Arapinis2021quantumphysical}.

The security measure is described through an authentication protocol where a verifier attempts to authenticate the identity of an honest prover. As shown in 
\cref{fig:Selective Unforgeability}, the authenticatoin protocol is structured into two phases: an \emph{enrollment phase} and a \emph{verification phase}.

\begin{figure}
    \centering
    \includegraphics[scale= 0.4]{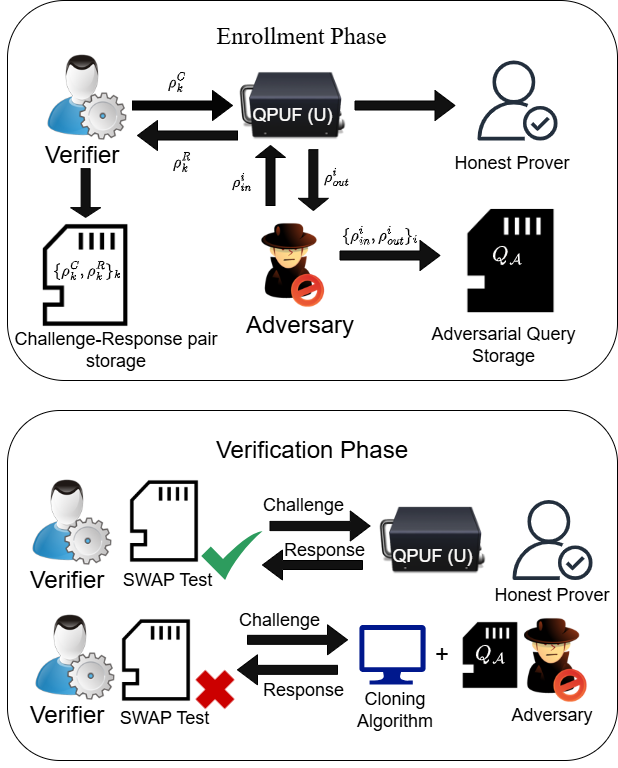}
    \caption{Selective Unforgeability. In the \emph{enrollment phase}, the verifier and the adversary collect and store information from the QPUF device by querying it. Then, the QPUF device is handed over to the honest prover. In the \emph{verification phase}, the verifier sends a challenge state and the adversary attempts to send back the correct response. The validity of the response is measured by a SWAP Test\cite{SWAP_Mina}.}
    \label{fig:Selective Unforgeability}
\end{figure}

During the \emph{enrollment phase}, the verifier prepares two copies of $M$ quantum input (\textit{challenge}) states, expressed as, 
\begin{equation}
    C \coloneqq \{\left(\rho_{k}^{C}\right)^{\otimes2} = U_k\ketbra{0}{0}U_{k}^{\dagger}\}_{k \in \mathbb{Z}_{M}},
\end{equation}
where each $U_k$ is independently sampled according to the Haar measure on the unitary group. The verifier applies the QPUF $U$ on one copy of these states $\{\rho_{k}^{C}\}_{k\in\mathbb{Z}_{M}}$, obtaining the output (\textit{response}) states, 
\begin{equation}
\label{eq: Sel Res}
    R \coloneqq \{\rho_{k}^{R} = U U_k\ketbra{0}{0}U_{k}^{\dagger}U^{\dagger}\}_{k \in \mathbb{Z}_{M}},
\end{equation}
and stores these states with their corresponding copy of the challenge states, forming a \emph{challenge-response} database $CRP$,
\begin{equation}
    CRP \coloneqq \{\left(\rho_{k}^{C},\rho_{k}^{R}\right)\}_{k \in \mathbb{Z}_{M}}.
\end{equation}
 
In the \emph{verification phase}, the verifier sends a challenge state from $CRP$ to the honest prover. The prover uses the QPUF to generate $M$ response states, represented as $\{\rho^P_k\}_{k \in \mathbb{Z}_{M}}$. Subsequently, the verifier conducts $M$ independent SWAP tests \cite{Arapinis2021quantumphysical,SWAP_Mina} between the prover's response states and the stored response states in $CRP$. The verification is accepted if and only if all SWAP tests are passed; otherwise, it is rejected.

\begin{figure*}
    \centering
    \includegraphics[scale = 0.8]{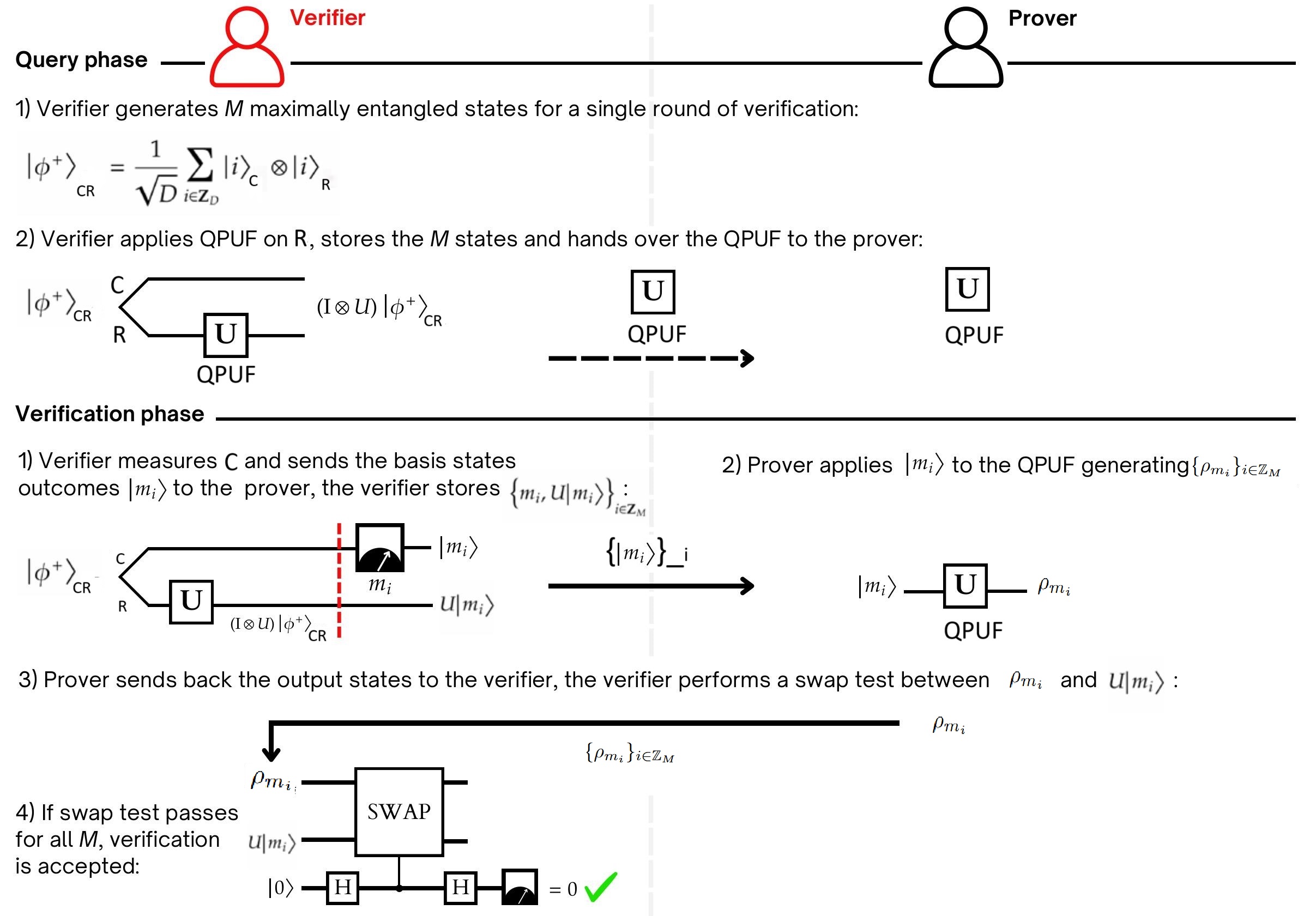}
    \caption{Single Round verification with $M$ trials. The query phases describes how the verifier stores $M$ Choi states of the QPUF $U$ and transfers the device to the prover. The verification phase describes, the challenge generation via measurements and subsequent verification with SWAP test. This figure is taken from \cite{Ghosh2024}.}
    \label{fig:MB-QPUF}
\end{figure*}

\subsection{Security}
\label{Section: QPT}
Let $\mathcal{A}$ be an arbitrary adversary with query access to a QPUF described by a unitary operation $U$. Such an adversary can construct a query database given by:
\begin{equation}
Q_{\mathcal{A}}(U) \coloneqq \{\rho^{i}_{\textrm{in}}, \rho^{i}_{\textrm{out}} \equiv U\rho^{i}_{\textrm{in}}U^{\dagger}\}_{i \in \mathbb{Z}_{|Q{\mathcal{A}}|}}.
\end{equation}
Given a security parameter \cite{Ghosh2024} $\lambda$, we classify the adversary according to the size of the query set. If $|Q_{\mathcal{A}}(U)| = \text{poly}(\lambda)$, the adversary is termed a Quantum Polynomial-Time (QPT) adversary; if $|Q_{\mathcal{A}}(U)| = \exp(\lambda)$, it is a Quantum Exponential-Time (QET) adversary. (The security parameter is typically considered to be the number of qudits used in the QPUF system.)

For a given challenge state $\rho^{C}_{k}$, let $P(Q_{\mathcal{A}}(U),\rho^{C}_{k})$ denote the adversary's probability of producing a state $\rho^{*}$ that passes the SWAP test with the correct response state $\rho^{R}_{k}$. This leads us to the definition of \textit{selective unforgeability}:

\begin{definition}[Selective Unforgeability]
A unitary QPUF $U$ is \emph{selectively unforgeable} if, for every given challenge state $\rho_{k}^C \in C$, the probability $P(Q_{\mathcal{A}}(U),\rho^{C}_{k})$ of generating the correct response state $\rho_{k}^R \in R$ for an adversary is negligible in the security parameter $\lambda$.
\end{definition}
In \cite{Arapinis2021quantumphysical}, it was proven that any Haar unitary QPUF is \textit{selectively unforgeable}.

\section{Brief review of Measurement Based QPUF (MB-QPUF) Selective Unforgeability scheme \cite{Ghosh2024}}
\label{SEC:III}

The MB-QPUF scheme was introduced in \cite{Ghosh2024}. For clarity, we restate it in a form more convenient for our analysis. Following the approach proposed in \cite{Arapinis2021quantumphysical}, the unitary QPUF is modeled as a Haar-random unitary \( U \). As shown in \cref{fig:MB-QPUF}, the MB-QPUF protocol consists of two distinct phases: the \emph{query phase} and the \emph{verification phase}.

\subsection{Verifier's Query Phase:}
The data collection of the verifier is described in the following steps:

\begin{itemize}
    \item Given $N,M = poly(\lambda)$(\textit{lowest degree of the polynomial is at least $1$}), the verifier initialises $NM$ many maximally entangled states,$\ket{\Phi^{+}}_{CR}$, on two $D-$dimensional systems labelled $C,R$, where,
    \begin{equation}
    \ket{\Phi^{+}}_{CR} \equiv \frac{1}{\sqrt{D}}\sum_{i \in \mathbb{Z}_{D}} \ket{i}_{C} \otimes \ket{i}_{R}. 
    \end{equation}
    \item Then they apply the QPUF unitary channel $U$ locally on system $R$ of all the maximally entangled states to obtain the following states,
    \begin{equation}
    \left(\mathbb{I} \otimes U \right)\ket{\Phi^{+}}_{CR} = \frac{1}{\sqrt{D}}\sum_{i \in \mathbb{Z}_{D}} \ket{i}_{C} \otimes U \ket{i}_{R} 
    \end{equation}
    \item Finally they store these states and hand over the QPUF $U$ to the prover.
\end{itemize}
System $C$ is used for challenge generation, while system $R$ is employed for response verification, as detailed in the next subsection.

\subsection{Prover's Verification}
We outline a single round of verification. The verifier performs $M$ measurements on the challenge-generation system $C$, producing outcomes $\{m_{i}\}_{i \in \mathbb{Z}_{M}}$ that are communicated to the prover. This measurement step collapses the quantum states in the response system $R$ into states $\{U\ket{m_i}\}_{i \in \mathbb{Z}_{M}}$.

In response, the prover sends back $M$ quantum states $\{\rho_{m_i}\}_{i\in \mathbb{Z}_{M}}$. The verifier then executes $M$ independent SWAP tests \cite{SWAP_Mina}, each comparing a state from system $R$ with the corresponding response state received from the prover. Verification succeeds only if all $M$ SWAP tests pass; otherwise, it is rejected.

\subsection{Security}
Let $\mathcal{A}$ be any arbitrary adversary with query access to the QPUF $U$. Then they would be able to create a query database,
\begin{equation}
    Q_{\mathcal{A}} \coloneqq \{\rho^{i}_{in}, \rho^{i}_{out} \equiv U\rho^{i}_{in}U^{\dagger}\}_{i \in \mathbb{Z}_{|Q_{\mathcal{A}}|}}. \notag
\end{equation}
Let the verifier's measurement outcomes be denoted as \(\{m_{i}\}_{i \in \mathbb{Z}_{M}}\). The winning strategy of the Adversary for generating states that passes all the $M$ rounds of the SWAP test with the authentic response states, could in general be an entangled state over $M$ systems. However, in \cite{SWAP_Mina} it was proven that the best strategy for the adversary would be to send separable states. Therefore, let the adversary prepare $M$ many guess states $\{\tilde{\rho^{i}}_{\mathcal{A}}\}_{i \in \mathbb{Z}_{M}}$.
We can then write the overall winning probability of the adversary for a single round of verification,
\begin{align}
\label{eq: MB-QPUF prob}
    P_{\mathcal{A}} (U) = \prod_{i \in \mathbb{Z}_{M}} \left(\frac{1}{2} + \frac{1}{2}\left| \Tr[\tilde{\rho^{i}}_{\mathcal{A}}U\ketbra{m_{i}}U^{\dagger}] \right|^{2} \right) 
\end{align}
With this setup, \textit{measurement selective unforgeability} was defined and proved for Haar unitary QPUF in \cite{Ghosh2024}. 

\begin{definition}[Measurement Selective Unforgeability]
  A unitary QPUF $U$, when used in the MB-QPUF model described above, is defined as \textit{measurement selective unforgeable} if,
  for any QPT adversary with query access to the QPUF, the overall success probability \( P_{\mathcal{A}}(U) \) for a single round of verification is at most negligible in the security parameter.
\end{definition}
\cref{Theorem 1}, proven in \cite{Ghosh2024}, establishes that Haar-random unitary QPUFs are \textit{measurement selective unforgeable}.

\begin{theorem}[Measurement Selective Unforgeability \cite{Ghosh2024}]
\label{Theorem 1}
For any security parameter \(\lambda\) and the number of trials \(M = poly(\lambda)\), in a single verification round, the expected success probability of any adversary is bounded by:
\[
\mathbb{E}_{U \sim Haar}[P_{\mathcal{A}}(U)] \leq \frac{1}{2^M} + \negl(\lambda) = \negl(\lambda).
\]
Consequently, any MB-QPUF scheme is measurement-selective unforgeable.
\end{theorem}

\begin{remark}
    The function $\negl(\lambda)$ denotes any class of functions that decay  exponentially in $\lambda$, i.e. $\mathcal{O}(\frac{1}{exp(\lambda)})$. 
\end{remark}

\section{Chaotic Hamiltonian based QPUF}
\label{SEC:IV}
Previous security models for unitary QPUFs \cite{Arapinis2021quantumphysical, Ghosh2024} rely on the assumption that the QPUF unitary is drawn from the Haar-random ensemble. However, implementing such ensembles on NISQ devices remains a significant challenge.  

To address this, we propose an alternative approach where the QPUF unitary channel is generated via the time evolution of a quantum chaotic Hamiltonian \( H \):  
\begin{equation}
   U(t) = e^{-iHt}. 
\end{equation}
Here, \( t \) represents the processing time of the QPUF device, which can either be fixed for a given family of QPUFs or publicly known to an adversary. It is straightforward to show that any unitary generated by this construction satisfies all the hardware requirements outlined in page 7 of~\cite{Arapinis2021quantumphysical}. The \emph{unknownness} requirement follows from the randomness of the Hamiltonian \( H \), while \emph{collision resistance} and \emph{robustness} follow from the sensitive dependence on initial conditions of chaotic dynamics.

To analyze this construction, we model the chaotic Hamiltonian as a random Hermitian matrix \( H \) drawn from the Gaussian Unitary Ensemble (GUE) \cite{beenakker1997random,hu2022hamiltonian} in a Hilbert space of dimension \( D = d^{\lambda} \), following the probability distribution:  
\begin{equation}
P(H) \propto \exp\left(-\frac{D}{2} \mathrm{Tr} H^2\right).
\end{equation}  
The energy scale is normalized such that the spectral density(i.e. joint probability density of the eigenvalues) of \( H \) follows the Wigner semicircle law \cite{beenakker1997random}, with spectral radius $2$, in the large \( D \) limit:  
\begin{equation}
\rho(E) = \frac{1}{2\pi} \sqrt{4 - E^2}, \quad E \in [-2,2].
\end{equation}  
The spectral $2$-point correlation function (i.e. marginal density over remaining $D-2$ eigenvalues) between the $p^{th}$ and $q^{th}$ eigen energies is given by \cite{cotler2017chaos},
\begin{align}
    &\rho^{(2)}(E_{p},E_{q}) \notag\\ &= \frac{D^{2}}{D(D-1)}\left(\rho(E_p)\rho(E_q) - \frac{\sin^{2}(D(E_p - E_q))}{(D\pi(E_p - E_q))^2} \right)
\end{align}
It can be easily seen that in the large $D$ limit, we have,
\begin{equation}
    \rho^{(2)}(E_p, E_q) \approx \rho(E_p)\rho(E_q). \label{eq: eigen independence}
\end{equation}
This implies that the eigenvalues become approximately pairwise uncorrelated in large dimensions. As we will demonstrate later, this property will play a crucial role in the security proof.

The unitary time evolution generated by \( H \) admits the eigen-decomposition:  
\begin{equation}
U(t) = e^{-iHt} = V \Lambda(t) V^{\dagger}, \label{eq: eigen decomp}
\end{equation}  
where \( V \) is the unitary eigenvector matrix of \( H \), and \( \Lambda(t) \) is the diagonal matrix,
\begin{equation}
\label{eq: Lambda eigen}
\Lambda(t)_{pq} = e^{-iE_p t}\delta_{pq},    
\end{equation}
with $E_p$ being the $p^{th}$ eigen energy of $H$ (working with standard unit where $\hbar = 1$). For a GUE Hamiltonian $H$, the matrix $V$ is Haar random and is independently distributed from the matrix $\Lambda(t)$. With this setup, we will now proceed to present the security results for our proposed model in the following theorems.
\begin{theorem}[Measurement Selective Unforgeability]
Upon replacing the Haar-random unitary in the MB-QPUF scheme \cite{Ghosh2024} with a chaotic QPUF unitary with evolution time $t = \lambda >> 1$, the expected probability of success for an adversary is bounded by:
\begin{align}
    \mathbb{E}_{U \sim \textrm{GUE}}(P_{\mathcal{A}}(U)) &\leq \frac{1}{2^M} + \left( \frac{1}{t -1} \right)\cdot \mathcal{O}(\frac{1}{D - |Q_{\mathcal{A}}|}), \notag \\
    &= \frac{1}{2^M} + \left( \frac{1}{\lambda -1} \right)\cdot \mathcal{O}(\frac{1}{q^{\lambda} - |Q_{\mathcal{A}}|}), \notag \\
    &= \frac{1}{2^M} + \negl(\lambda)
\end{align}
where \( \lambda \) is the security parameter and  \( M = \text{poly}(\lambda) \) is the number of trials in a single verification round.
\end{theorem}

\begin{proof}
Following the arguments used in the proof of Theorem 1 in \cite{Ghosh2024}, it suffices to show that the quantity
\begin{equation}
\mathbb{E}_{U}\left[\Tr\left(\rho^{i}_{\mathcal{A}}U\ketbra{m_{i}}U^{\dagger}\right)\right]
\end{equation}
is negligible in the security parameter \(\lambda\). 

Recall the eigen-decomposition of the unitary \(U\) from \eqref{eq: eigen decomp}:
\begin{equation}
U(t) = e^{-iHt} = V \Lambda(t) V^{\dagger},
\end{equation}
where \(V\) is a Haar-random unitary matrix and \(\Lambda(t)\) is a diagonal unitary matrix containing eigenvalues of \(U\). Crucially, \(V\) and \(\Lambda(t)\) are independently distributed \cite{Pseudo2024}.

To analyze the information leaked about \(U\) through the adversary’s query dataset \(Q_{\mathcal{A}}\), we separate the analysis into contributions from \(V\) and \(\Lambda(t)\), exploiting their independence.

First, we analyze the amount of information about \( V \) that the adversary can gain from their queries. Given that the adversary's query set \(Q_{\mathcal{A}}\) has at most \(|Q_{\mathcal{A}}| \ll D\) elements, the queried states span at most a \(|Q_{\mathcal{A}}|\)-dimensional subspace. As the columns of \(V\) form a complete orthonormal basis, we conservatively assume the adversary obtains complete knowledge of the subspace spanned by up to \(|Q_{\mathcal{A}}|\) columns of \(V\). Following arguments from Theorem 3 in \cite{Ghosh2024}, the remaining \(D - |Q_{\mathcal{A}}|\) columns of \(V\) form a Haar-random unitary \(W\) acting on the complementary \((D - |Q_{\mathcal{A}}|)\)-dimensional subspace. Thus, adversarial queries effectively reduce \(V\) to a Haar-random unitary \(W\) of dimension \(D - |Q_{\mathcal{A}}|\).

Finally, we analyze the amount of information about \( \Lambda(t) \) that the adversary can gain from their queries. Similarly, assuming the adversary knows all eigenvalues corresponding to the collected columns of \(V\), the diagonal eigenvalue matrix \(\Lambda(t)\) effectively reduces to a smaller diagonal matrix \(\Lambda'(t)\) of dimension \(D - |Q_{\mathcal{A}}|\), described by a suitable joint density over the remaining eigenvalues.

Consequently, the distribution of the QPUF unitary \(U\) effectively reduces to a \((D - |Q_{\mathcal{A}}|)\)-dimensional random unitary \(S\), explicitly represented as:
\begin{equation}
    S = W \Lambda'(t) W^{\dagger},
\end{equation}
where \(W\) is a Haar-random unitary and \(\Lambda'(t)\) is the modified diagonal eigenvalue matrix. Substituting the eigen decomposition of $S$ from above in place of $U$, we have 
   \begin{align}
&\mathbb{E}_{U}\Tr[\rho^{i}_{\mathcal{A}}U\ketbra{m_{i}}U^{\dagger}] \notag \\
&\leq \mathbb{E}_{W,\Lambda'(t)} \Tr[W\Lambda'(t)W^{\dagger} \ketbra{m_i}{m_i}W\Lambda^{'\dagger}(t)W^{\dagger} \rho_{\mathcal{A}}^{i}]
   \end{align}
Without loss of generality and for notational simplicity, let the measurement outcome be \( m_i = k \in \mathbb{Z}_{D} \), and denote \(\rho_{\mathcal{A}}^{i}\) by \(\rho\). Thus, have,
\begin{align}
&\mathbb{E}_{U}\Tr[\rho^{i}_{\mathcal{A}}U\ketbra{m_{i}}U^{\dagger}] \notag \\ 
&\leq \mathbb{E}_{W,\Lambda'(t)} \Tr[W\Lambda'(t)W^{\dagger} \ketbra{k}{k}W\Lambda^{'\dagger}(t)W^{\dagger} \rho] \notag\\
&= \mathbb{E}_{W,\Lambda'(t)} \left[\sum W_{ij} \Lambda'(t)_{jn}\Bar{W}_{mn}\delta_{mk}W_{kl}\Lambda^{'\dagger}(t)_{lq}\Bar{W}_{pq}\rho_{pi}\right] \notag \\
&= \mathbb{E}_{\Lambda'(t)} \Bigg[\sum \left( \int W_{ij}W_{kl}\Bar{W}_{mn}\Bar{W}_{pq}
 dW \right) \delta_{mk}  \times \notag \\ 
&\hspace{1.5cm}\Lambda'(t)_{jn}\Lambda^{'\dagger}(t)_{lq} \rho_{pi} \Bigg]
\end{align}
where $\Bar{W}$ denotes complex conjugation of the matrix $W$. The integral in parentheses can be expressed using standard results involving Weingarten functions \cite{weingarten1978asymptotic}. By defining \(\tilde{D} = D - |Q_{\mathcal{A}}|\) and evaluating this integral explicitly, we obtain the following bound:
\begin{align}
&\mathbb{E}_{U}\Tr[\rho^{i}_{\mathcal{A}}U\ketbra{m_{i}}U^{\dagger}] \notag \\
&\leq \frac{1}{\Tilde{D}^{2} - 1} \left(\mathbb{E}_{\Lambda'(t)} \left[\sum_{pq} e^{i(E_{p} - E_q)t} \right] \right) + \frac{\Tilde{D}}{\Tilde{D}^{2} - 1} \label{eq: pre-bound}
\end{align}
It is well-established \cite{beenakker1997random,cotler2017chaos} that the eigenvalues of a Hamiltonian drawn from GUE exhibit \emph{level repulsion}, meaning that for any pair of eigenvalues \( E_p, E_q \), we certainly have \( E_p \neq E_q \). Consequently, as noted in \cite{beenakker1997random}, the eigenvalues can be unambiguously ordered along the real line. This ordering is illustrated in Fig. \ref{fig:level spacing}.

\begin{figure}
    \centering
    \includegraphics[scale = 0.48]{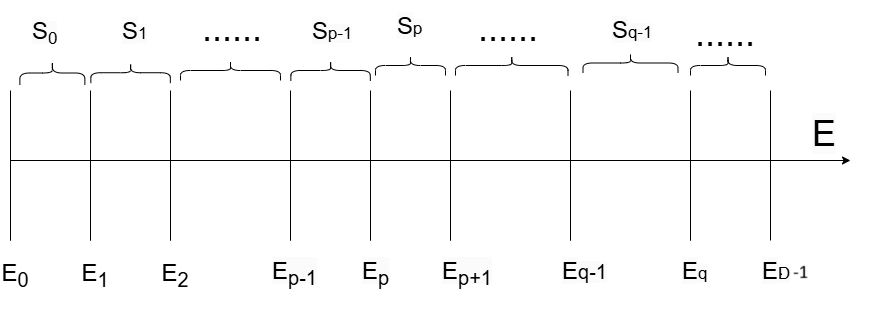}
    \caption{Level spacing statistics. The spacing between the \( p^{\text{th}} \) and \( (p+1)^{\text{th}} \) energy levels is denoted by \( S_p \).}
    \label{fig:level spacing}
\end{figure}

We define the level spacing between consecutive eigenvalues as \( S_p := E_{p+1} - E_p \). Hence, for any pair \( p < q \), the energy difference can be written as a sum of level spacings:
\[
E_q - E_p = \sum_{\alpha = p}^{q-1} S_{\alpha}.
\]
Let \( \gamma \) denote the average spacing between adjacent energy levels. We then define the normalized spacings \( s_\alpha := \frac{S_\alpha}{\gamma} \), whose distribution is governed by the Wigner surmise \cite{beenakker1997random}:
\[
P(s_\alpha) \propto s_\alpha e^{-\frac{\pi s_\alpha^2}{4}}.
\]

Without loss of generality, we now set the mean level spacing \( \gamma = 1 \), which will simplify the expressions. As previously demonstrated in \eqref{eq: eigen independence}, the eigenvalues become asymptotically uncorrelated in the large $\tilde{D}$ limit. This approximate independence implies that the nearest-neighbor spacings \( \{S_\alpha\} \) are also weakly correlated. Thus, the joint distribution of spacings approximately factorizes:
\[
P(S_0, S_1, \ldots, S_{\tilde{D}-1}) \approx \prod_{\alpha = 0}^{\tilde{D}-1} P(S_\alpha).
\]
We now compute the expectation
\[
\mathbb{E}_{\Lambda'(t)} \left[\sum_{p,q} e^{i(E_p - E_q)t} \right],
\]
by rewriting the energy differences in terms of level spacings:
\begin{align}
&\mathbb{E}_{\Lambda'(t)} \left[\sum_{p,q} e^{i(E_p - E_q)t} \right], \notag \\
&= \mathbb{E}_{\Lambda'(t)} \left[ \sum_{p< q} e^{i(E_p - E_q)t} + e^{-i(E_p - E_q)t} + \sum_{p = q} e^{i(E_p - E_q)t} \right], \notag \\ 
&= \mathbb{E}_{\{S_\alpha\}} \left[\sum_{p < q} e^{i\left(\sum_{\alpha = p}^{q-1} S_\alpha\right)t} +e^{-i\left(\sum_{\alpha = p}^{q-1} S_\alpha\right)t}\right] + 0, \notag \\ &[\because \Pr(E_p = E_q ) = 0, \textrm{due to level repulsion \cite{beenakker1997random}}] \notag \\ 
&= \sum_{p<q} \prod_{\alpha = p}^{q-1} \int_{0}^{\infty} (e^{i S_\alpha t}+ e^{-i S_\alpha t}) P(S_\alpha)\, dS_\alpha, \notag \\
&=  \sum_{p<q} \prod_{\alpha = p}^{q-1} \int_{0}^{\infty} \cos(S_{\alpha}t) P(S_\alpha) dS_\alpha.
\end{align}
Noting that the variable $S_{\alpha}$ will be integrated out, we let,
\begin{equation}
    \int_{0}^{\infty} \cos(S_{\alpha}t) P(S_\alpha) dS_\alpha = \chi (t),
\end{equation}
we get,
\begin{align}
    \mathbb{E}_{\Lambda'(t)} \left[\sum_{p,q} e^{i(E_p - E_q)t} \right] = \sum_{p < q} \chi^{q-p}(t)
\end{align}
For $t > 0$, we can bound $\chi(t) \leq \frac{2}{t}$, (see \cref{sec: Appendix}, \cref{eq:chi_time_bound}). Plugging this bound, we have,
\begin{align}
     \mathbb{E}_{\Lambda'(t)} \left[\sum_{p,q} e^{i(E_p - E_q)t} \right] &= \sum_{(q-p) \equiv d = 1}^{\tilde{D} - 1} (\tilde{D} - d) \frac{1}{t^d} \notag \\
     &\leq \tilde{D} \sum_{d = 1}^{\infty} \frac{1}{t^d} \notag \\
     &= \frac{\tilde{D}}{t - 1}
\end{align}
Finally, we have,
\begin{align}
&\mathbb{E}_{U}\Tr[\rho^{i}_{\mathcal{A}}U\ketbra{m_{i}}U^{\dagger}] \notag \\
&\leq \left( \frac{1}{t-1} \right)\cdot \left(\frac{2\tilde{D}}{\tilde{D}^{2} - 1} \right) \notag \\ 
&= \left( \frac{1}{\lambda -1} \right)\cdot \mathcal{O}(\frac{1}{\tilde{D}}), \; [\textrm{setting}\; t = \lambda \geq 2] \notag \\
&= \negl(\lambda)
\end{align}

This completes the proof.
\end{proof}
The proof of the above theorem shows that the evolution time $t$ can be made public and it can scale linearly in the security parameter $\lambda$. We further establish that the chaotic Hamiltonian QPUF is also \textit{selectively unforgeable}.
\begin{theorem}(Selective Unforgeability Chaotic QPUF)
    Any chaotic unitary QPUF is selectively unforgeable.
\end{theorem}
\begin{proof}
Consider an arbitrary response state \( \rho_{\mathrm{sel}} \) that the adversary aims to guess in a selective unforgeability scheme. From \cref{eq: Sel Res}, without loss of generality, we have,
\begin{equation}
    \rho_{\mathrm{sel}} = U(t)\, U_k \ketbra{0}{0} U_k^\dagger\, U(t)^\dagger,
\end{equation}
where \( U(t) \) denotes the chaotic QPUF unitary and \( U_k \sim \mu_{\mathrm{Haar}} \) is a unitary drawn from the Haar measure on \( \mathrm{U}(D) \).

By the left invariance of the Haar measure, the composition \( U(t) U_k \) is distributed identically to \( U_k \), i.e., \( U(t) U_k \sim U_k \). Thus, the distribution of \( \rho_{\mathrm{sel}} \) is independent of \( U(t) \), and \( \rho_{\mathrm{sel}} \sim U_k \ketbra{0}{0} U_k^\dagger \).

Let the adversary's query database have size \( |Q_{\mathcal{A}}| \), and let \( \rho_{\mathcal{A}} \) denote the adversary’s guess state. Then, by Theorem 3 of \cite{Ghosh2024}, we have the following upper bound on the expected overlap:
\begin{equation}
    \mathbb{E}_{\rho_{\mathrm{sel}}} \left[\mathrm{Tr}\left(\rho_{\mathrm{sel}} \rho_{\mathcal{A}}\right)\right] \leq \frac{1}{D - |Q_{\mathcal{A}}|}.
\end{equation}
as $D = d^{\lambda}$ and $|Q_{\mathcal{A}}| = poly(\lambda)$, the probability decays exponentially with the security parameter $\lambda$. 
This concludes the proof.
\end{proof}

In the following sections, we will discuss two methods for implementing our QPUF model.

\begin{figure*}[t!]
    \centering
    \begin{subfigure}[t]{0.5\textwidth}
        \centering
    
    \includegraphics[scale = 0.29]{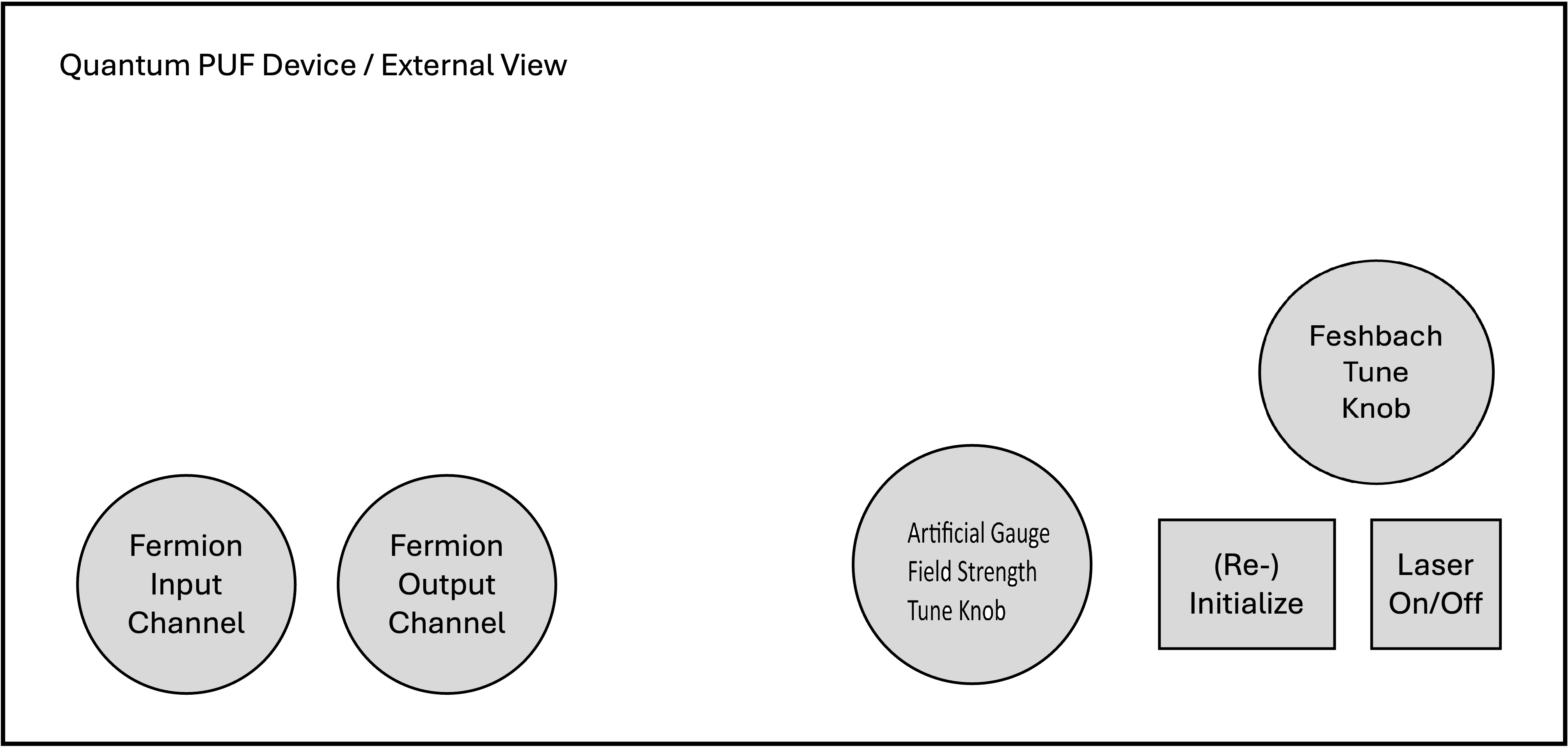}
        \caption{External view. Fermionic input/output channels are shown, enabling particle exchange. The gauge field tuning knob adjusts the system to induce positive or complex hopping parameters. The Feshbach tuning knob controls the interaction strength of the impurities. The laser activation button generates the Kagome lattice structure, while the re-initialize button resets the device for a new application, effectively changing the device as a new QPUF.}
        \label{fig: External}
    \end{subfigure}%
    ~
    \begin{subfigure}[t]{0.5\textwidth}
        \centering
        \includegraphics[scale = 0.29]{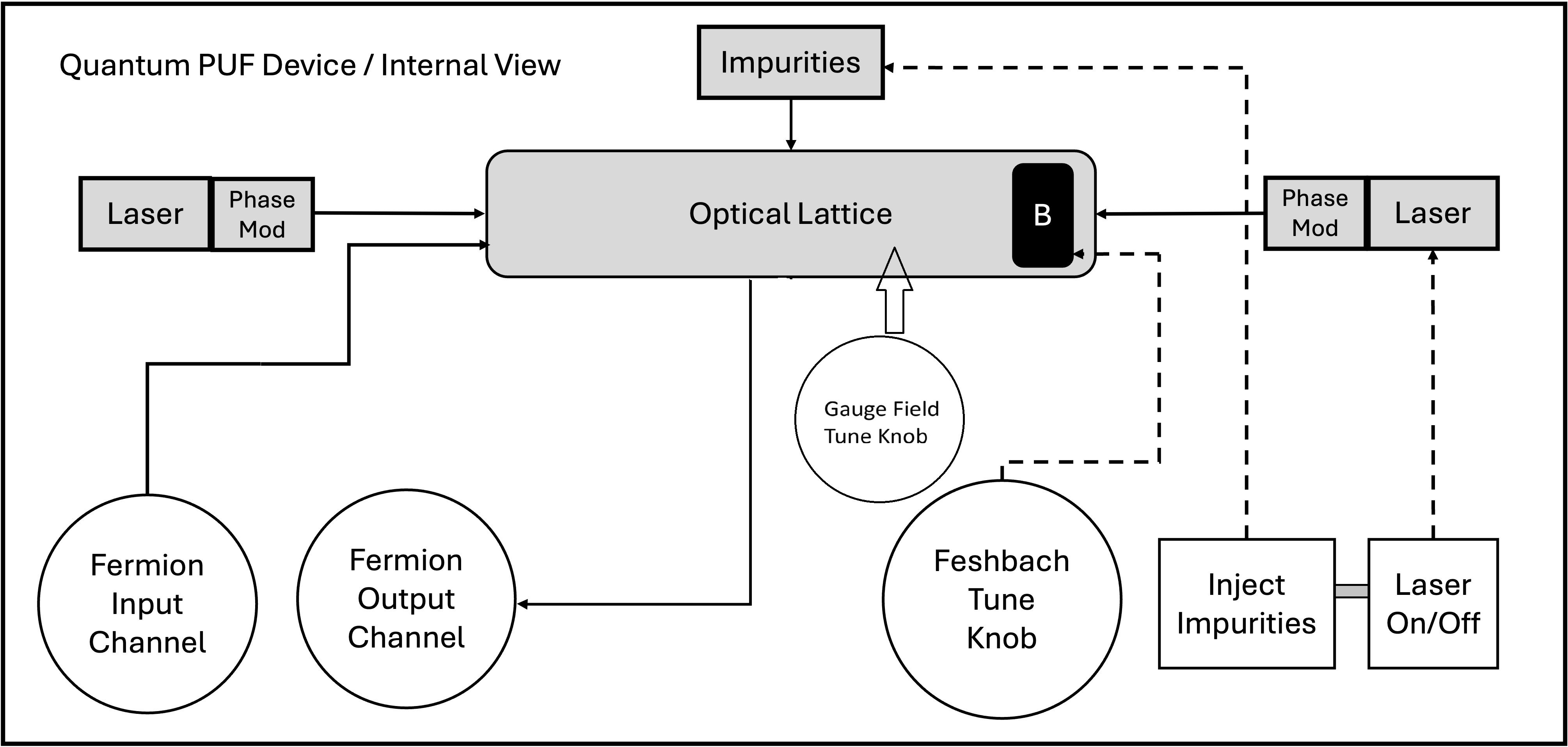}
        \caption{Internal View. The laser on/off switch traps the fermions in a Kagome lattice. It is connected with the impurity injector to create a disordered medium. After the time evlution is completed, the switch is turned off and the fermions are collected from the output channel.}
        \label{fig: Internal}
    \end{subfigure}
    \caption{QPUF device architecture. The left panel displays the external view, highlighting user controls and interface elements, while the right panel shows the internal view, detailing the underlying operational components.}
    \label{Fig: QPUF Device Architecture}
\end{figure*}

\section{Physical Hamiltonian Construction}
\label{Sec:V}
A concrete physical realization of the QPUF device can be modeled using the Sachdev-Ye-Kitaev (SYK) Hamiltonian ($ \mathcal{H}_{\textrm{SYK}}$)\cite{sachdevChaos, kitaeChaos}, which describes a system of randomly interacting fermions:
\begin{equation}
    \mathcal{H}_{\textrm{SYK}} = - \mu \sum_{i} c_{i}^{\dagger}c_{i} + \sum_{i>j , k>l} J_{ijkl} c^{\dagger}_{i}c^{\dagger}_{j}c_{k}c_{l},
\end{equation}
where \( c_i^\dagger \) and \( c_i \) are fermionic creation and annihilation operators acting on a Hilbert space of dimension $D = d^n$ for a $n$-qudit system of qudit dimension $d$. The coefficients \( J_{ijkl} \) are complex-valued random couplings drawn from a Gaussian distribution with zero mean and variance \( \langle |J_{ijkl}|^2 \rangle = \frac{J^2}{2d^3} \). The parameter \( \mu \) denotes the chemical potential.

In this construction, the fermionic modes act as the computational degrees of freedom (qubits or qudits). The QPUF protocol begins with initializing these fermions in a fixed input state. The system then undergoes unitary time evolution under \( \mathcal{H}_{\textrm{SYK}} \), generating an output state that is effectively random due to the underlying disorder in the couplings.

Importantly, prior work \cite{OpticalCavity_Chaos} has shown that by tuning the chemical potential \( \mu \), the system exhibits a phase transition between chaotic and non-chaotic regimes. To ensure unpredictability and cryptographic hardness, an honest user can calibrate \( \mu \) (for e.g. with external magnetic field) to maintain the Hamiltonian in the chaotic phase, thus guaranteeing strong pseudorandomness in the output state, even when the QPUF device is sourced from an untrusted manufacturer. 

The SYK model has been successfully simulated in laboratory settings using nuclear spins \cite{SYK_Nuc} and ultracold atoms in optical lattices \cite{SYK_Optical}, demonstrating the practical feasibility of our proposed QPUF construction.

\subsection{Device Architecture Schematic}

In this subsection, we outline a high-level QPUF device architecture inspired by the experimental realization presented in \cite{SYK_Optical}, which serves as a proof of concept.

A central challenge in implementing the SYK model is its inherently non-local interactions. To emulate SYK-type physics within a spatially local and physically realizable system, the authors of \cite{SYK_Optical} propose an optical Kagome lattice Hamiltonian with nearest-neighbor interactions and randomly distributed impurities across lattice sites. The impurities create a delta-function type potential barrier.

They demonstrate that this model effectively reduces to a SYK Hamiltonian. In other words, they show (refer to eq.7 in \cite{SYK_Nuc}), that the system dynamics can be described by an effective Hamiltonain $\mathcal{H}_{eff}$, which is given by,
\begin{equation}
    \mathcal{H}_{eff} = (2t - \mu) \sum_{i} c_i^{\dagger}c_{i} + \sum_{i>j, k>l} J_{ijkl} c_{i}^{\dagger} c_{j}^{\dagger}c_k c_l
\end{equation}
which is exactly the SYK Hamiltonian. While each experimental component of this setup is now standard, integrating them into a single platform may still require additional effort.

The critical components are as follows:
\begin{enumerate}
    \item \textbf{Random on-site potentials:} Realized by heavy atoms randomly loaded into the optical lattice. Their interaction strength can be tuned using Feshbach resonances \cite{hara2014three}.
    \item \textbf{Controlled hopping (positive or complex):} 
    Achieved through artificial gauge fields generated by either (a) a zero-averaged homogeneous inertial force \cite{struck2012tunable} or (b) Raman-assisted tunneling in an asymmetric Kagome lattice \cite{aidelsburger2011experimental}.
\end{enumerate}
By tuning the strengths of the Feshbach resonances and the artificial gauge fields, the system can be driven into the chaotic phase, as demonstrated in \cite{SYK_Optical}.

Combining all components, we present the schematic architecture of the QPUF device in \cref{Fig: QPUF Device Architecture}. 

The external layout(see \cref{fig: External}) includes the following components:
\begin{itemize}
    \item \textbf{Fermionic qudit input–output channel:} Facilitates the injection and collection of fermionic qudits into and out of the device.
    \item \textbf{Feshbach resonance control:} A tuning knob for adjusting the impurity potentials by modifying interaction strengths via Feshbach resonances.
    \item \textbf{Artificial gauge field controller:} Regulates the hopping amplitude of fermionic qudits by tuning the strength of synthetic gauge fields.
    \item \textbf{Reinitialization button:} Resets all field strengths to predefined initialization values. This feature enables the same device to be reconfigured for use as multiple distinct QPUFs, supporting multiple users.
    \item \textbf{Laser on/off switch:} Activates the formation of the optical Kagome lattice when fermionic qudits enter the device. When switched off, the lattice is turned off, allowing fermions to propagate freely and be collected at the output.
\end{itemize}
The internal layout(see \cref{fig: Internal}) of the device comprises:
\begin{itemize}
    \item \textbf{Laser array with phase modulators}: Generates the optical Kagome lattice by interfering laser beams with tunable phase control.
    \item \textbf{Impurity injectors}: Introduce heavy atoms at random lattice sites to create local impurity potentials. Upon completion of the evolution, the impurities can be reabsorbed into the injectors. These injectors are synchronized with the laser switch(see \cref{fig: Internal}), ensuring that impurities are only introduced when the lattice is active.
\end{itemize}

\begin{figure}[t]
\centering
\begin{tikzpicture}
\begin{axis}[
    width=\columnwidth,        
    height=0.55\columnwidth,   
    xlabel={$p$ (Probability)},
    ylabel={$\lambda = -\log_2(p)$},
    xmode=log,
    log basis x=10,
    grid=both,
    minor grid style={gray!25},
    major grid style={gray!50},
    enlargelimits=true,
    xtick={0.001, 0.01, 0.1, 1},
    ytick={0,2,4,6,8,10,12}
]
\addplot[
    domain=0.001:1,
    samples=200,
    thick,
    blue
]
{-log2(x)};
\end{axis}
\end{tikzpicture}
\caption{Qubit case: Number of qubits of the QPUF system ($\lambda$) as a function of the adversary’s success probability \(p\) in forging the QPUF, where \(\lambda = -\log_2(p)\).}
\label{Fig:Resource}
\end{figure}
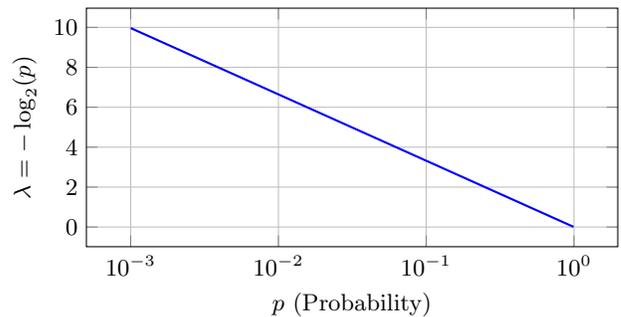

\subsection{Resource Estimate}

In this section, we provide a preliminary estimate of the resources required to construct a secure QPUF based on the SYK Hamiltonian. Resources can be broadly classified into space and time costs. In this work, we focus solely on estimating the space cost, while the time cost depends heavily on the specifics of the implementation. We leave its characterization as an open problem for future experimental realizations.

From the security result in \cref{Thm:GUE security}, we know that for a single round of verification, the adversary’s success probability $p$ is bounded by

$$
    p \leq \frac{1}{2^{M}} + \mathrm{negl}(\lambda),
$$

where the negligible term $\negl(\lambda)$ scales inversely with the system's dimension. For a system composed of $\lambda$ qubits, this can be approximated as

$$
    p \leq \frac{1}{2^{M}} + \frac{1}{2^{\lambda}}.
$$

Since $M$ can be somewhat chosen arbitrarily to suppress the first term, the dominant factor becomes the $\lambda$-dependent term. This term governs the scalability of the QPUF and serves as a rough estimate of the quantum resources needed. Naturally, increasing the number of qubits makes the system more secure but also more challenging to implement.

In \cref{Fig: Resource}, we plot the adversary’s success probability as a function of the number of qubits $\lambda$, offering a visual representation of the resource requirements. While each component of the proposed device architecture relies on standard experimental techniques, integrating them into a unified platform remains to be done and may require additional efforts and incur additional costs. Consequently, we currently lack a concrete estimate for the maximum number of qubits our architecture can support. Although the architecture is theoretically scalable, understanding its practical limits remains an open problem.


\section{Pseudo-chaotic Hamiltonian Simulation}
\label{Sec:VI}

Pseudo-chaotic Hamiltonian ensembles were first introduced and rigorously defined in \cite{Pseudo2024}. For clarity, we restate it in a form more convenient for our analysis. Before defining a pseudo-chaotic Hamiltonian ensemble, we first clarify what it means to have black-box query access to a Hamiltonian $H$. This refers to the ability to collect input-output pairs from its time evolution \( e^{-iHt} \). We now proceed to define the pseudo-chaotic Hamiltonian.

\begin{definition}[Pseudo-Chaotic Hamiltonian Ensemble]
Let \( \mathcal{E}^{D}_{\mathrm{GUE}} \) denote an ensemble of chaotic Hamiltonians drawn from the Gaussian Unitary Ensemble (GUE) of dimension \( D \). A \emph{pseudo-chaotic} Hamiltonian ensemble \( \mathcal{E}^{D}_{\mathrm{PGUE}} \) is any ensemble that satisfies the following two conditions:
\begin{itemize}
    \item It does not exhibit any characteristics of quantum chaos noted in \cref{table:chaos_probes}.
    \item It is computationally indistinguishable from \( \mathcal{E}^{D}_{\mathrm{GUE}} \) by any efficient algorithm \( \mathcal{A}^H \) with only black-box access to the Hamiltonian. In other words, if \( \mathcal{A}^H \) is any algorithm that checks if an Hamiltonian is GUE or not, that is, outputs $1$, if it is GUE and $0$ if not, then we have,
    \begin{align}
        &\left| \Pr_{H \sim \mathcal{E}^{D}_{\mathrm{GUE}}}[A^H(\cdot) = 1] 
        - \Pr_{H \sim \mathcal{E}^{D}_{\mathrm{PGUE}}}[A^H(\cdot) = 1] \right| \notag \\ 
        &= \negl(n).
    \end{align}
\end{itemize}
\end{definition}
By definition, this gives a valid QPUF construction against quantum polynomial-time (QPT) adversaries with only query access. The chaotic and pseudo-chaotic Hamiltonian ensembles are computationally indistinguishable based on query access. As a result, their corresponding time evolution operators are also indistinguishable to any efficient algorithm.

\begin{table*}
\centering
\begin{tabular}{|c|c|}
\hline
\textbf{Probe to chaos} & \textbf{Definition} \\
\hline
4-point OTOC  & $\text{OTOC}(H,t) := \frac{1}{d} \text{tr}[O_1(t) O_2 O_1(t) O_2]$ \\
\hline
2-Rényi entropy  & $S_{A|B}(H,t;\rho_0) := -\log \text{tr}_A\left[(\text{tr}_B(e^{-iHt} \rho_0 e^{iHt}))^2\right]$ \\
\hline
Operator entanglement  & $\text{LOE}(H,t) := S_2(\text{tr}_B B', |O_1(t)\rangle \langle O_1(t)|)$ \\
\hline
Stabilizer entropy & $M_\alpha(H,t;\rho_0) := \frac{1}{1-\alpha} \log \frac{1}{d} \sum_P \text{tr}^{2\alpha}\left(P e^{-iHt} \rho_0 e^{iHt}\right)$ \\
\hline
\end{tabular}
\caption{Probes to chaos considered in \cite{Pseudo2024}. Chaos is typically associated with systems for which there exists times $t = O(\text{poly} \, n)$ such that the 2-Rényi entropy, operator entanglement, and stabilizer entropy are extensive in system size, and the 4-point OTOC is exponentially vanishing in system size.}
\label{table:chaos_probes}
\end{table*}

An efficient algorithm for constructing a pseudo-chaotic Hamiltonian ensemble was proposed in \cite{Pseudo2024}. For completeness, we briefly summarize the key steps of the construction.

Any GUE Hamiltonian \( H_{\mathrm{GUE}} \) admits the spectral decomposition:
\begin{equation}
    H_{\mathrm{GUE}} = U D U^\dagger,
\end{equation}
where \( U \) is a Haar-random unitary matrix and \( D \) is a diagonal matrix containing its eigenvalues. Let the total dimension of the Hamiltonian be \( D \).

The pseudo-chaotic construction proceeds as follows:
\begin{itemize}
    \item Sample \( U \) from a unitary \( t \)-design, which serves as an efficient approximation to a Haar-random unitary.
    \item Sample \( d = \mathrm{poly}(\lambda) \) eigenvalues independently from the Wigner semicircle distribution \cite{Pseudo2024}, where \( \lambda \) is an efficiency parameter.
    \item Construct the diagonal matrix \( D \) by repeating each sampled eigenvalue \( \frac{D}{d} \) times, making the spectrum highly degenerate.
\end{itemize}

This construction is computationally efficient for suitable values of \( \lambda \), since both \( t \)-design sampling and eigenvalue generation require only polynomial resources in \( \lambda \).

\section{Discussion and Future Research}
\label{Sec:VII}
We have shown that chaotic Hamiltonians enable secure QPUF constructions. In particular, we have proved that QPUFs based on chaotic Hamiltonians satisfy both \emph{measurement-selective unforgeability} \cite{Ghosh2024} and \emph{selective unforgeability} \cite{Arapinis2021quantumphysical}.

While the SYK model offers a compelling theoretical foundation, its experimental realization remains challenging due to its inherently non-local, all-to-all interactions. This raises important research questions: To what extent can such non-locality be practically implemented—i.e., how many qubits or qudits can a realistic SYK model support? How do we introduce fault tolerance into the picture? How would the processing time of the QPUF device scale with number of qudits? Moreover, since SYK is maximally chaotic, can we design more local, less chaotic Hamiltonians that still enable secure QPUF constructions?

Despite these challenges, SYK models have been simulated in physical systems such as nuclear spins~\cite{SYK_Nuc} and cold atom optical lattices~\cite{SYK_Optical}, paving the way for practical, chaos-based QPUF architectures.

Physical implementations are also critical for realistic security modeling. Most existing QPUF works assume black-box query access, but real-world adversaries may have physical access to the device. SYK-based QPUFs on physical platforms allow us to study these stronger adversarial models.

This highlights an important trade-off between physical implementation and pseudo-chaotic simulation. While physically constructing chaotic Hamiltonians is more challenging, it allows us to model stronger adversaries with different kinds of physical access, beyond the standard black-box query model. 

The pseudo-chaotic construction also raises an intriguing open question: since randomness arises both from the sampling of the \( t \)-design and from the eigenvalue distribution, does this additional layer of randomness offer any meaningful advantage?

Another promising direction is to compare how different physical platforms—such as cold atoms, superconducting qubits, or trapped ions—affect QPUF security and performance. Understanding these differences can guide practical implementations tailored to specific application needs.

\section*{Appendix}
\label{sec: Appendix}
Here we present the bound on the temporal form factor $\chi(t)$, for completeness, We recall that,
\begin{equation}
    \chi(t) \;=\; \int_{0}^{\infty} s\, e^{-\frac{\pi}{4}s^{2}} \cos(st)\, ds.
    \label{eq:chi_def}
\end{equation}
Setting \(a = \frac{\pi}{4}\), we write
\begin{equation}
    \chi(t) \;=\; \Re \int_{0}^{\infty} s\, e^{-a s^{2}} e^{i s t}\, ds.
\end{equation}
Applying integration by parts with 
\(u = s e^{-a s^{2}}\) and \(dv = e^{i s t} ds\),
we obtain
\begin{equation}
    \chi(t)
    = -\frac{1}{t} \int_{0}^{\infty} 
      (1 - 2 a s^{2})\, e^{-a s^{2}} \sin(st)\, ds,
\end{equation}
where the first term of the integration by-parts vanishes because \(s e^{-a s^{2}} \to 0\) as \(s \to \infty\).
Taking absolute values and using \(|\sin(st)| \le 1\) and \(|1 - 2 a s^{2}| \le 1 + 2 a s^{2}\),
we find
\begin{align}
    |\chi(t)|
    &\le \frac{1}{|t|} \int_{0}^{\infty} (1 + 2 a s^{2})\, e^{-a s^{2}} ds \notag \\
    &= \frac{1}{|t|}\, \sqrt{\frac{\pi}{a}}.
\end{align}
For \(a = \pi/4\), this yields the simple time–dependent bound
\begin{equation}
        |\chi(t)| \;\le\; \frac{2}{|t|}, \qquad t \neq 0.
    \label{eq:chi_time_bound}
\end{equation}

\bibliographystyle{IEEEtran}
\bibliography{bibliography}
\end{document}